\documentclass{article}
\usepackage[utf8]{inputenc} % allow utf-8 input
\usepackage{amsmath, amsthm, amssymb,amsfonts}
\usepackage{xcolor}
\newtheorem{prop}{Proposition}
\newtheorem{lem}{Lemma}
\newtheorem{thm}{Theorem}
\newtheorem{remark}{Remark}

\usepackage{algpseudocode,algorithm,algorithmicx}  
\usepackage{graphicx}

\title{Neural Network Gradient Hamiltonian Monte Carlo}

\date{}
\author{Lingge Li, Andrew Holbrook, Babak Shahbaba, Pierre Baldi\\
University of California, Irvine}

\begin{document}

\maketitle

\begin{abstract}

%The Bayesian paradigm is especially useful for scientific modeling of complex or hierarchical phenomena. Big Bayesian models are difficult to compute, but the Hamiltonian Monte Carlo algorithm enables inference for high-dimensional models with thousands of parameters.  Big models seem to be in direct conflict with big data---the same algorithm scales notoriously poorly with the number of data points.  We show that, by approximating the log-posterior gradient with a multi-output neural network, Hamiltonian Monte Carlo can handle both big data \emph{and} big models. The proposed neural network gradient Hamiltonian Monte Carlo achieves an up to 40-fold increase in efficiency when compared to its top competitors, all the while maintaining asymptotically exact inference.

Hamiltonian Monte Carlo is a widely used algorithm for sampling from posterior distributions of complex Bayesian models. It can efficiently explore high-dimensional parameter spaces guided by simulated Hamiltonian flows. However, the algorithm requires repeated gradient calculations, and these computations become increasingly burdensome as data sets scale. We present a method to substantially reduce the computation burden by using a neural network to approximate the gradient. First, we prove that the proposed method still maintains convergence to the true distribution though the approximated gradient no longer comes from a Hamiltonian system. Second, we conduct experiments on synthetic examples and real data to validate the proposed method.

\end{abstract}

\section{Introduction}

%Rooted in a probabilistic framework, Bayesian models offer natural and consistent ways for parameter inference, model checking and decision making. However, in most settings, Bayesian models produce posterior distributions that are analytically intractable. 

%Inference on intractable posterior distributions is one of the main challenges in Bayesian modeling, and Markov chain Monte Carlo (MCMC) is the standard approach to sample from such posterior distributions. Given the rapidly growing complexity of both models and data, efficient MCMC algorithms are imperative for the wider application of Bayesian models. The ideal algorithm must handle both the high dimensionality introduced by complex models and the huge computational cost associated with big data.

%Simple MCMC algorithms such as the Metropolis-Hastings (MH) random walk do not work well in high-dimensional parameter spaces as they ``blindly wander.'' In contrast, 
Hamiltonian Monte Carlo (HMC) uses local geometric information provided by the log-posterior gradient to explore the high posterior density regions of the parameter space \cite{neal2011mcmc}. Compared to the Metropolis-Hastings random walk algorithm, HMC has high acceptance probability and low sample auto-correlation even when the parameter space is high-dimensional. That said, the advantages of HMC come at a computational cost that limits its application to smaller data sets. The gradient calculation involves the entire data set and scales linearly with the number of observations. As HMC needs to calculate the gradient multiple times within every single step, performing HMC on millions of observations requires an enormous computational budget. Allowing HMC to scale to large data sets would help tackle the double challenge of big data and big models.

There have been two main approaches to scaling HMC to larger data sets. The first is stochastic gradient HMC, which calculates the gradient on subsets of the data. \cite{welling2011bayesian} implemented a stochastic gradient version of Langevin Dynamics, which may be viewed as single-step HMC. \cite{chen2014stochastic} introduced stochastic gradient HMC with ``friction'' to counterbalance the inherently noisy gradient. However, these methods may not be optimal because subsampling substantially reduces the acceptance probability of HMC \cite{betancourt2015fundamental}. %As a result, these methods have trouble scaling to large dimensions.

The second approach relies on a surrogate function, the gradient of which is less expensive to calculate. 
%The naive method that approximates the gradient piecewise on a grid severely suffers from the curse of dimensionality. 
\cite{rasmussen2003gaussian}\cite{lan2016emulation} used Gaussian process (GP) to produce satisfactory results in lower dimensions. However, training a GP is itself computationally expensive and training points must be chosen with great care. More recently, \cite{zhang2015hamiltonian} implemented neural network surrogate with random bases. It outperforms GP surrogate in their experiments but fails in parameter spaces of moderate dimensionality.

In this paper, we develop a third approach, neural network gradient HMC (NNgHMC), by using a neural network to directly approximate the gradient instead of using it as a surrogate. We also train all the neural network weights through backpropagation rather than having random weights \cite{zhang2015hamiltonian}. Compared to existing methods, our proposed approach can emulate Hamiltonian flows accurately even when dimensionality increases. In Section 3, details of our method and proof of convergence are presented. Section 4 includes experiments to validate our method and comparisons with previous methods on synthetic and real data.

\section{Background}

\subsection{Hamiltonian Monte Carlo}
Let $x\sim \pi(x|q)$ denote a probabilistic model with $q$ as its corresponding parameter. We also make $q$ a random variable by giving the parameter a prior distribution $\pi(q)$. The integration constant of the posterior distribution
\begin{align}
 \pi(q|x)=\frac{\pi(x|q)\pi(q) }{\int \pi(x|q)\pi(q) \, dq} 
\end{align}
 is usually analytically intractable, but the distribution can be numerically simulated using MCMC. The Metropolis-Hastings algorithm constructs a Markov chain that randomly proposes a new state $q'$ from current state $q$ based on transition distribution $g(q'|q)$ and moves from $q$ to $q'$ with probability $\min\{1, \frac{\pi(q'|X)g(q|q')}{\pi(q|X)g(q'|q)}\}$. Unfortunately, in a higher dimensional space, the probability of randomly moving to $q'$ drops dramatically. Therefore, the MH algorithm has trouble exploring the posterior efficiently in higher dimensions.

The idea of HMC is to explore a frictionless landscape induced by potential energy function $U$ and kinetic energy function $K$ where potential energy $U(q)=-\log \pi(x|q)\pi(q)$ is proportional to the negative log posterior. HMC introduces an auxiliary Gaussian momentum $p$, and $K(p)$ is the negative log density of $p$. Potential energy $U$ tends to convert to kinetic energy $K$ so $q$ will likely move to a position with higher posterior density. More formally, the Hamiltonian system is defined by the following equations.
\begin{align}
H(q, p) &= U(q) + K(p) = - \Big(\log \pi(q)+ \sum_{i=1}^N \log \pi(x_i|q) \Big) + \frac{1}{2} p^Tp\, ,\\ 
\frac{dq}{dt}&=\frac{\partial H}{\partial p}=\frac{\partial K}{\partial p}=p\label{eq::theta_update} \\ 
\frac{dp}{dt}&=-\frac{\partial H}{\partial q}=-\frac{\partial U}{\partial q}=  \nabla_q \Big(\log \pi(q)+ \sum_{i=1}^N \log \pi(x_i|q) \Big)\label{eq::p_update}  \, .
\end{align}
In theory, convergence of HMC is guaranteed by the time reversibility of the Hamiltonian dynamics which, in turn, renders the Markov chain transitions reversible, thus ensuring detailed balance.   By conservation of the Hamiltonian, HMC has acceptance probability 1 and can travel arbitrarily long trajectories along energy level contours. In practice, the Hamiltonian dynamics is simulated with the leapfrog algorithm which adds numerical errors. To ensure convergence to the posterior, a Metropolis correction step is used at the end of each trajectory. 

Within each simulated trajectory, the leapfrog algorithm iterates back and forth between Equations \eqref{eq::theta_update} and \eqref{eq::p_update}, the latter of which features a summation over the log-likelihood evaluated at \emph{each} separate data point. For large data sets, this repeated evaluation of the gradient becomes infeasible. In Section \ref{NNgHMC}, we show how to greatly speed up HMC using neural network approximations to this gradient term, but first we introduce an important predecessor to our method, the surrogate HMC class of algorithms.

\subsection{Surrogate HMC}

%\begin{figure}
%\centering
%\includegraphics[width=0.4\linewidth]{shallow}
%\includegraphics[width=0.4\linewidth]{gradient}
%\caption{The second network computes the gradient of the first network.}
%\end{figure}

Two methods for approximating the log-posterior in the HMC context have already been advanced.  The first uses a Gaussian Process regression to model the log-posterior, the second uses a neural network. We refer to the latter as neural network surrogate HMC (NNsHMC). It is natural that both models would be used in such a capacity: Cybenko \cite{cybenko1989approximation} showed that neural networks can provide universal function approximation, and Neal \cite{neal2012bayesian} showed that certain probabilistic neural networks converge to Gaussian processes as the number of hidden units goes to infinity. In this section, we focus on NNsHMC, since it is more closely related to our method (Section \ref{NNgHMC}).

NNsHMC approximates the potential energy $U$ with a neural network surrogate $\widehat{U}$ and uses $\nabla \widehat{U}$ during leapfrog steps. The surrogate neural network has one hidden layer with softplus activation: 
\begin{align}
\widehat{U}(q)=W_2\ln(1+\exp(W_1q+b_1))+b_1
\end{align}
where $W_1,W_2$ and $b_1,b_2$ are weight matrices and bias vectors, respectively. Under this formulation, one can explicitly calculate the gradient
\begin{align}
\nabla\widehat{U}=W_1^Tdiag(W_2)\frac{1}{1+\exp{(-(W_1q+b_1))}}
\end{align}
and represent $\nabla\widehat{U}$ with another neural network, which is just the backpropagation graph of $\widehat{U}$. Therefore, we can view neural network surrogate as using a constrained network with tied weights and local connections to approximate the gradient. 

For training the neural network, Zhang et al \cite{zhang2015hamiltonian} uses extreme learning machine (ELM) \cite{huang2004extreme}. ELM is a simple algorithm that randomly projects the input to the hidden layer and only trains the weights from the hidden layer to the output. Random projection is widely used in machine learning but backpropagation is the ``default" training method for most neural networks with its optimality theoretically explained by Baldi et al \cite{baldi2016theory}. Moreover, since the goal is to improve computational efficiency, we want to make the surrogate neural network as small as possible. From this point of view, large hidden layers often seen in ELMs are less than optimal.

\section{Neural network gradient HMC}\label{NNgHMC}

\begin{algorithm}
\caption{Neural network gradient HMC}
\begin{algorithmic}
\State Initialize $q^{(0)}$, leapfrog step number $L$ and step size $\epsilon$
\For{$t = 1,2,...,T$}
\State $q_0=q^{(t-1)}$
\State Sample momentum $p_0\sim \mathcal{N}(0, I)$
\State $p_0=p_0-\frac{\epsilon}{2}\widehat{\nabla U}(q_t)$ \Comment{Leapfrog steps with approximated gradient $\widehat{\nabla U}$ instead of $\nabla U$}
\For{$i = 1,2,...,L$}
\State $q_i=q_{i-1}+\epsilon p_{i-1}$
\State $p_i=p_{i-1}-\epsilon \widehat{\nabla U}(q_i)$
\EndFor
\State $p_L=p_L-\frac{\epsilon}{2}\widehat{\nabla U}(q_L)$
%\If{sample}  
%\State Use neural network $G$ to approximate the gradient:
%\State $p_0=p_0-\frac{\epsilon}{2}G(q_t)$
%\For{$i = 1,2,...,L$}
%\State $q_i=q_{i-1}+\epsilon p_{i-1}$
%\State $p_i=p_{i-1}-\epsilon G(q_i)$
%\EndFor
%\State $p_L=p_L-\frac{\epsilon}{2}G(q_L)$
%\EndIf
\State $r=\exp{(H(q_L,p_L)-H(q_0, p_0))}$, $u\sim Uniform(0,1)$
\If{$u<\min(1,r)$} \Comment{Metropolis accept/reject based on $H=U+K$}
\State $q^{(t)}=q_L$
\Else
\State $q^{(t)}=q_0$
\EndIf
\EndFor  
\end{algorithmic}
\end{algorithm}

In contrast to previous work, NNgHMC does not use a surrogate function for $U$ but fits a neural network to approximate $\nabla U$ directly with backpropagation. Training data $(q,\nabla U(q))$ for the neural network are collected during the early period of HMC shortly after convergence. Once the approximate gradient is learned, the algorithm is exactly the same as classical HMC, but with neural network gradient $\widehat{\nabla U}$ replacing $\nabla U$.  Details are given in Algorithm 1.

One benefit of our method occurs as early as the data collection process.  Since we approximate the gradient $\nabla U$ and not $U$, we can collect more training data faster: surrogate HMC must reach the end of a leapfrog trajectory before obtaining a single (functional evaluation) training sample; the same leapfrog trajectory renders a new (gradient evaluation) training sample for each leapfrog step, and the number of such steps in a single trajectory can range into the hundreds. 

Suppose that there are $N$ data points $x_n$ and that the parameter space is $d$-dimensional.  In this case, gradient calculations involve $d$ partial derivatives
\begin{align}
 \frac{\partial U}{\partial q_j} = -\frac{\partial}{\partial q_j} \log \Big( \pi(q) \prod_{i=1}^N \pi(x_i|q)\Big) = -\frac{\partial}{\partial q_j} \log \pi(q) - \sum_{i=1}^N \frac{\partial}{\partial q_j} \log \pi(x_i|q)  \, ,
\end{align}
each of which involves a summation over the $N$ data points. On the other hand, performing a forward pass in a shallow neural network is proportional only to the hidden layer size $s\ll N$. Once the neural network is trained on burn-in samples, posterior sampling with approximated gradient is orders of magnitude faster.

Although the neural network gradient approximation $\widehat{\nabla U}(q)$ is not the same as $\nabla U(q)$, the method nonetheless samples from the true posterior. If one were able to simulate the Hamiltonian system directly, i.e. without numerical integration, then all the benefits of HMC would be preserved in the limit, as the gradient field approximates the true gradient field to arbitrary degree. On the other hand, the NNgHMC transition kernel---characterized by the approximate gradient leapfrog integrator combined with the Metropolis accept-reject step---leaves the posterior distribution invariant.  We formalize the relevant results here and defer proofs to the appendix.

An important litmus test for the validity of our method is that it should leave the Hamiltonian invariant in the limit as step-sizes and gradient approximation errors approach zero. In turn, this result will imply high acceptance probabilities when the system is simulated from numerically, and when gradient approximations are good.
\begin{prop}
When the system induced by the approximate gradient field is simulated directly, changes in the Hamiltonian $H(q,p)=U(q) +K(p)$ converge in probability to 0 as the approximate gradient converges pointwise to the true gradient. That is, for a sequence of approximate gradient fields $\{\widehat{\nabla_q^n U}\}_{n=1}^\infty$ converging to the true gradient field $\nabla_qU$, the change in Hamiltonian values satisfies
\begin{align}
\Big(\frac{dH}{dt}\Big)_n = o_p(1) \, .
\end{align}
\end{prop}

\begin{proof}
Following \cite{cybenko1989approximation}, assume we are able to construct a sequence of approximate gradients $\widehat{\nabla^n_q H}$ satisfying
\begin{align}
\nabla_q H = \widehat{\nabla^n_q H} + E_n(q), \quad E_n(q) \in B_{1/n}(0) \, ,
\end{align}
where $B_{1/n}(0)$ is the ball around the origin of radius $1/n$. In this case, the vector field given by the approximate gradient induces a new system of equations:  
\begin{align}
    \frac{d q_i}{dt} &= \frac{\partial H}{\partial p_i} \\ \nonumber
    \frac{d p_i}{dt} &= - \frac{\widehat{\partial H}}{\partial q_i}  \, .
\end{align}
 Then it follows that
\begin{align}
\frac{d H}{dt} &= \sum_{i=1}^d \Big[ \frac{d q_i}{dt} \frac{\partial H}{\partial q_i} + \frac{d p_i}{dt} \frac{\partial H}{\partial p_i} \Big] \\ \nonumber
&= \sum_{i=1}^d \Big[  \frac{\partial H}{\partial p_i} \big(\frac{\widehat{\partial H}}{\partial q_i} + E_{n,i}(q)\big) - \frac{\widehat{\partial H}}{\partial q_i} \frac{\partial H}{\partial p_i} \Big] \\ \nonumber
&=  \sum_{i=1}^d \frac{\partial H}{\partial p_i}E_{n,i} \\ \nonumber
&= p^TE_n \sim N(0, E^T_n E_n) \, .
\end{align}
This last line implies $p^TE_n$ is $O_p(\sqrt{E_n^TE_n})$, and hence that $\frac{d H}{dt}$ is $o_p(1)$.
\end{proof}

We note that Proposition 1 is a local result, and that local deviations from the true Hamiltonian flow will accrue to larger global deviations in general.  While this may seem disconcerting, NNgHMC maintains remarkably high acceptance rates in practice. To help understand why this is the case, we present local and global error analyses for the dynamics of the ordinary differential equation initial value problem 
\begin{align}
    \frac{d}{dt}z = f(z) \, , \quad z(t_0) = z^0 \in \mathbb{R}^k \, ,
\end{align}
approximated with function $\widehat{f}\approx f$.  These results will then be related back to NNgHMC by specifying $z=(q,p)$ and  
\begin{align}
    z=(q,p)^T , \quad f(q,p) = \Big(p, - \frac{\partial H}{\partial q}\Big)^T ,\quad \mbox{and} \quad \widehat{f}(q,p) = \Big(p, - \frac{\widehat{\partial H}}{\partial q}\Big)^T  .
\end{align}
The general form of the following proofs follows after Section 2.1.2 of \cite{leimkuhler2004simulating}.

\begin{prop}
(Local error bounds) Let $z^0=z(0)$ be the initial value, let $z(\Delta t)$ be the value of the exact, true trajectory after traveling for time $\Delta t$, and let $z^1$ be the value of the computed trajectory using Euler's method applied to the approximated gradient field.  Finally, assume that the exact solution is twice continuously differentiable. Then the local error $\epsilon^1=z(\Delta t) - z^1$ has the following bounds:
\begin{align}
\|\epsilon^1\| \leq \Delta t \, \delta + O(\Delta t^2) \, ,
\end{align}
where $\delta = \| f(z_0) - \widehat{f}(z_0)\|$ measures the difference between the true, exact trajectory and the approximated trajectory at point $z_0$. 
\end{prop}
\begin{proof}
The proof follows from the Taylor expansion of both $z(\Delta t)$ and $z^1$:
\begin{align}
\epsilon^1 &= \big( z_0 + \Delta t \,  \dot{z}(0) + \frac{1}{2}\Delta t^2 \ddot{z}(\tau)   \big) - \big( z_0 + \Delta t \widehat{f}(z_0)  \big) \\ \nonumber
           &= \Delta t \big( f(z_0) - \widehat{f}(z_0) \big) + \frac{1}{2}\Delta t^2 \ddot{z}(\tau)\, ,
\end{align}
where $\tau \in [0,\Delta t]$.  The result follows from the triangular inequality.
\end{proof}

From the above result, it follows that the local error rate approaches the $O(\Delta t^2)$ error rate of Euler's method using the true gradient field as $\delta = \| f(z_0) - \widehat{f}(z_0)\| = \| \frac{\partial H}{\partial q}(z_0) - \frac{\widehat{\partial H}}{\partial q}(z_0)\|$ goes to 0.  The same approach can be used to obtain global error bounds.

\begin{prop}
(Global error bounds) We adopt the same notation as above with the addition of the error at iteration $n$, $\epsilon^n=z(n\Delta t)- z^n$, where $z^n$ is the value after $n$ Euler updates using the approximate gradient field. Also, let $t_n=n\Delta t$.  Again we assume that the exact solution is twice differentiable, and we further assume that it is Lipschitz with constant $L$. Then the following bounds on $\epsilon^n$ hold:
\begin{align}
\|\epsilon^n\| \leq \big(e^{n\Delta tL}-1   \big)\big(\frac{\delta}{L} + O(\Delta t)\big) \, , \quad \mbox{for} \quad \delta =
\max \|f(z(j\Delta t))-\widehat{f}(z^j)\|\, ,
\end{align}
and $j=0,\dots,n$.
\end{prop}

\begin{proof}
The proof proceeds by recursion.  Assume that we have obtained $\epsilon^n= z(t_n)-z^n$.  Letting $\tau \in [t_n, t_{n+1}]$, a Taylor's expansion gives:
\begin{align}
\epsilon^{n+1} &= \Big( z(t_n) + \Delta t \dot{z}(t_n) + \frac{1}{2}\Delta t^2 \ddot{z}(\tau)  \Big) - \Big(z^n + \Delta t \widehat{f}(z^n)  \Big) \\ \nonumber
               &= \Big( z(t_n) + \Delta t\, f(z(t_n)) + \frac{1}{2}\Delta t^2 \ddot{z}(\tau)  \Big) - \Big(z^n + \Delta t \widehat{f}(z^n) \Big) \\ \nonumber
               &= \big(z(t_n) -z^n\big) + \Delta t \big(f(z(t_n)) - \widehat{f}(z^n) \big) + \frac{1}{2}\Delta t^2 \ddot{z}(\tau) \, .
\end{align}
But $\ddot{z}$ is continuous by assumption, so we can bound $\ddot{z}$ on the closed interval $[t_n,t_{n+1}]$ by a constant $M$. Furthermore, the Lipschitz assumption combined with the triangle inequality give:
\begin{align}
\|\epsilon^{n+1}\| &\leq \|\epsilon^n\| + \Delta t \Big(\|f(z(t_n))-f(z^n)\| + \|f(z^n) - \widehat{f}(z^n)\|\Big) + \frac{\Delta t^2 M}{2} \\ \nonumber
                   &\leq (1 + \Delta t L) \|\epsilon^n\| + \Delta t\, \delta + \frac{\Delta t^2 M}{2} 
\end{align}
Next we make use of the following recursion relationship:
\begin{align}
    a_{n+1} \leq C\,a_n + D \Longrightarrow a_n \leq C^n \, a_0 + \frac{C^n - 1}{C - 1}D \,
\end{align}
 for $C=(1+\Delta t L)$ and $D=\Delta t\,\delta + \Delta t^2 M/2$. Noting that $a_0=\epsilon^0=0$ gives
 \begin{align}
     \|\epsilon^n\| \leq (e^{t_nL}-1)\big(\frac{\delta}{L} + \frac{\Delta t M}{2L} \big) \, ,
 \end{align}
 and the result follows.
\end{proof}
The above result suggests that the usual numerical error caused by a large Lipschitz constant $L$ can overpower the effects of gradient approximation error $\delta$.

The preservation of volume entailed by both the theoretical Hamiltonian flow and the leapfrog integrator is important for HMC. The latter fact implies there is no need for Jacobian corrections within the accept-reject step.  It turns out that the NNgHMC dynamics also preserve volume, both for direct and for leapfrog simulation. 
\begin{lem}
Both for infinitesimal and finite step sizes, the NNgHMC trajectory preserves volume.
\end{lem}

\begin{proof}
For the finite case, the leapfrog algorithm iterates between shear transformations and so preserves volume \cite{neal2011mcmc}. For the case of direct simulation, we use the fact that the Hamiltonian vector field induced by the approximate gradient field has zero divergence (Liouville's Theorem, \cite{neal2011mcmc}). We use the notation of Proposition 1, but drop the subscript $n$ for the sake of readability:
\begin{align}
\sum_{i=1}^d \Big[\frac{\partial}{\partial q_i}\frac{d q_i}{dt} + \frac{\partial }{\partial p_i}\frac{d p_i}{dt} \Big] &= \sum_{i=1}^d \Big[\frac{\partial}{\partial q_i}\frac{\partial H}{\partial p_i} - \frac{\partial }{\partial p_i}\frac{\widehat{\partial H}}{\partial q_i} \Big] \\ \nonumber
&= \sum_{i=1}^d \Big[\frac{\partial}{\partial q_i}\frac{\partial H}{\partial p_i} - \frac{\partial }{\partial p_i}\big(\frac{\partial H}{\partial q_i}-E_i\big) \Big]  \\ \nonumber
&= \sum_{i=1}^d \frac{\partial }{\partial p_i}E_i = 0 \, .
\end{align}
\end{proof}

Not only does the NNgHMC trajectory preserve volume, it is reversible as well. This easy fact is shown in the proof of Proposition 2.
\begin{thm}
The NNgHMC transition kernel leaves the canonical distribution $\exp \{ -H(q,p)\}$ invariant.
\end{thm}

\begin{proof}
Since leapfrog integration preserves volume and since the Metropolis acceptance probability is the same as for classical HMC, all we need to show is that the leapfrog integration is reversible. This fact follows in the exact same way as for HMC, despite the use of an approximate gradient field to generate the dynamics:
\begin{align}
p_i(t + \epsilon/2) &= p_i(t)  - (\epsilon/2)\, \frac{\widehat{\partial U}}{\partial q_i}(q(t)) \\ \nonumber
q_i(t + \epsilon) &= q_i(t) + \epsilon \, p_i(t + \epsilon/2) \\ \nonumber
p_i(t+\epsilon) &= p_i(t + \epsilon/2) - (\epsilon/2)\, \frac{\widehat{\partial U}}{\partial q_i}(q(t+ \epsilon)) \, .
\end{align}
These are the same equations as in \cite{neal2011mcmc} except with  $\frac{\widehat{\partial U}}{\partial q_i}$ replacing  $\frac{\partial U}{\partial q_i}$.  Hence, just as in \cite{neal2011mcmc}, the NNgHMC leapfrog equations are symmetric and thus reversible: to reverse a sequence of leapfrog dynamics, negate $p$, take the same number of steps, and negate $p$ again. It follows that the NNgHMC transition kernel leaves the canonical distribution invariant and is an asymptotically exact method for sampling from the posterior distribution.
\end{proof}

Regardless of the accuracy of neural network gradient approximation, following the leapfrog simulated Hamiltonian proposal scheme would recover the true posterior distribution when combined with Metropolis-Hastings correction. If the gradient approximation is ``bad,'' NNgHMC would break down to a random walk algorithm. If the gradient approximation is ``close enough," NNgHMC would behave just like standard HMC, operating on energy level contours at a fraction of the computation cost. \textit{The neural network gradient approximation can be controlled with two tuning parameters: hidden layer size $h$ and training time $t$, in addition to leapfrog steps $l$ and step size $s$.} The neural network architecture is fixed to have one hidden layer and the number of units has to be pre-determined. Neural network training time can be either set to some number of epochs or dependent on a stopping criterion (typically based on change in loss function between epochs). Since there is no noise (error) in the gradient, overfitting is not a concern; the hidden layer size and training time could be relatively large.

Given sufficient training data, the neural network will be able to accurately approximate the gradient field. The important question is: how much training data should be collected? To address this, we propose a training schedule that consists of a start point, an end point, and a rate for gradient data collection. For example, one may wish to run a HMC chain to draw 5000 samples in total. A training schedule could be training the neural network every 200 draws between the 400th and 1000th draws. After the neural network is trained each time, one would use the approximated gradient to sample for some iterations. If the acceptance probability is similar to that of standard HMC, one would stop the training schedule and complete the entire chain with NNgHMC. Otherwise, standard HMC would be used to sample the remaining draws. \textit{Since training the neural network and using it to sample is much cheaper computationally compared to standard HMC, the training schedule would add little overhead even if the neural network gradient approximation fails.} 

\begin{figure}[h!]
\centering
\includegraphics[width=0.6\linewidth]{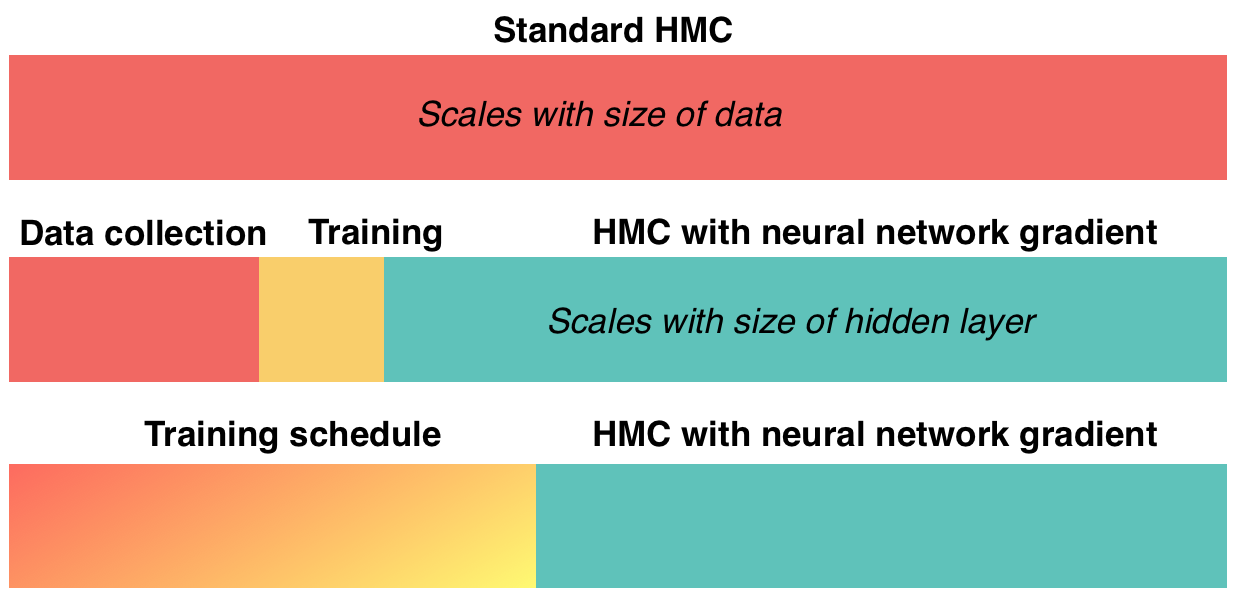}
\caption{After the neural network learns an accurate gradient approximation, the computation cost of sampling is substantially reduced compared to standard HMC. Therefore, the benefit of neural network gradient HMC depends on how much training data is enough for the neural network. Using a training schedule, we would stop standard HMC immediately after the neural network has learned from enough data.}
\label{fig:schedule}
\end{figure}

\section{Experiments}

In this section, we demonstrate the merits of proposed method: accuracy and scalability through a variety of experiments. The accuracy of gradient approximation can be reflected by high acceptance probability that is similar to standard HMC using the true gradient. Compared to draws from stochastic gradient HMC, the draws using proposed method are much more similar to standard HMC draws. Scalability means better performance when both data size $n$ and dimensionality $p$ increase. The performance metric is effective sample size (ESS) adjusted by CPU time. ESS estimates the number of ``independent'' samples by factoring $\rho(k)$ correlation between samples at lag $k$ into account:
\begin{align*}
    ESS=\frac{n}{1+2\sum_{k=1}^{\infty}\rho(k)}.
\end{align*}
The previous surrogate approach fails when $p$ reaches $40$ while the proposed method works well up to $p=200$. \textit{Lastly, speed evaluation is done on three real data sets and the proposed method consistently beats standard HMC even when the time to collect training data and train the neural network is included.} Our proposed NNgHMC method is implemented in Keras and uses the default Adam optimizer \cite{kingma2014adam} during training. All experiments are performed on a 3.4 GHz Intel Quad-Core CPU and our code is available at: github.com/linggeli7/hamiltonian.

\subsection{Distributions with challenging gradient fields} 

\textbf{The banana shaped distribution} in two dimensions can be sampled using the following un-normalized density

\begin{align}
    f(x_1, x_2)\propto\exp{-\frac{(Ax_1)^2}{200}-\frac{1}{2}(Cx_2+B(Ax_1)^2-100B)^2}
\end{align}
where $A,C$ control the scale in $x_1,x_2$-space and B determines the curvature. For HMC, the energy function is set to be $-\log{f(x_1,x_2)}$ and the its gradient can be easily calculated. Using leapfrog steps $l=5$ and step size $s=0.1$, standard HMC is used to sample 5000 draws with acceptance probability 0.58. Gradient values collected during the first 1000 draws are then used to train a neural network with hidden layer size $h=100$ for $t=50$ epochs. With the same tuning parameters, NNgHMC is used to sample 5000 draws with acceptance probability 0.57. Figure \ref{fig:banana} compares standard HMC and NNgHMC draws, the true and approximated gradient fields, and two long simulated leapfrog trajectories using both. The neural network learns the distorted gradient field accurately and NNgHMC completely recovers the banana shape.

\begin{figure}[h!]
\centering
\includegraphics[width=0.4\linewidth]{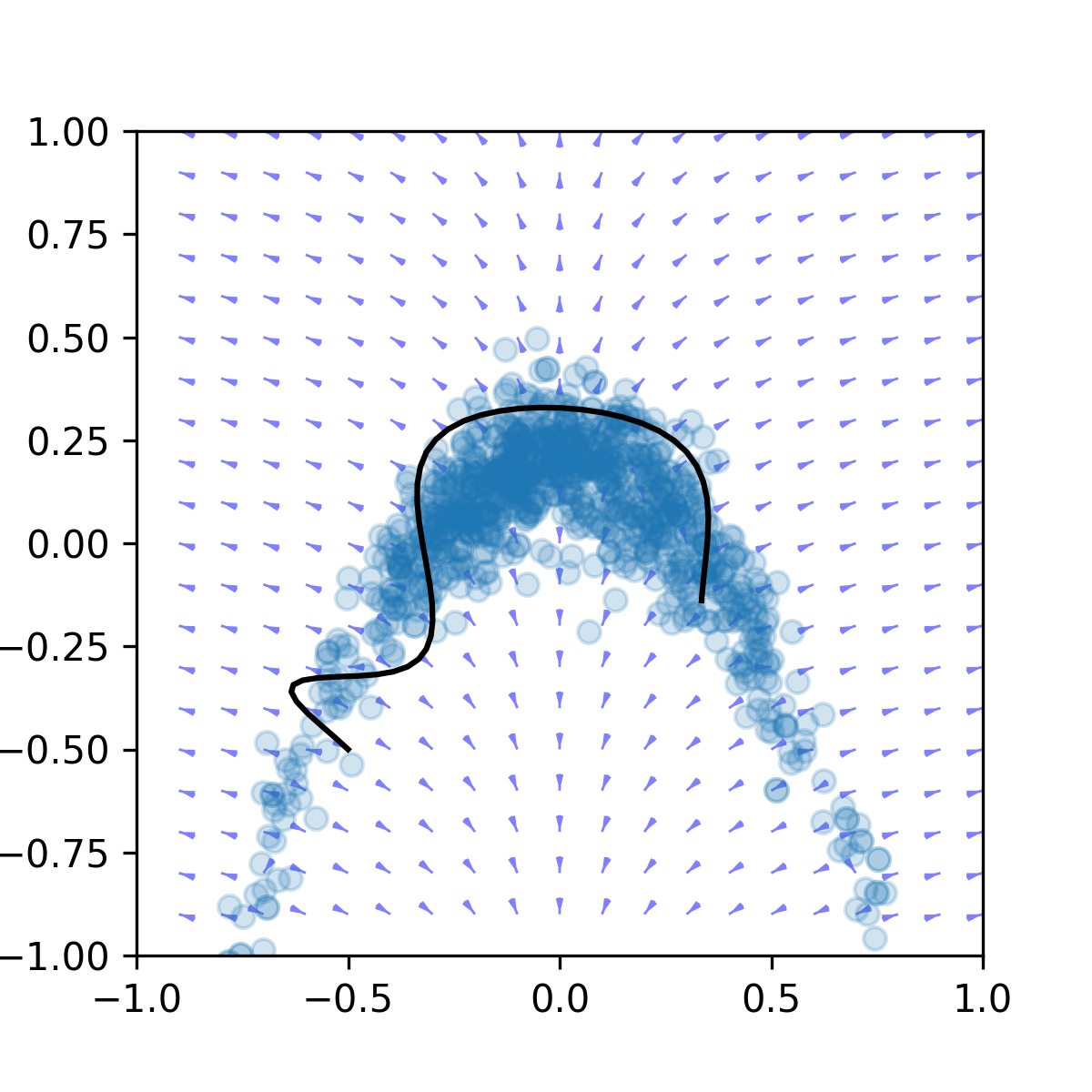}
\includegraphics[width=0.4\linewidth]{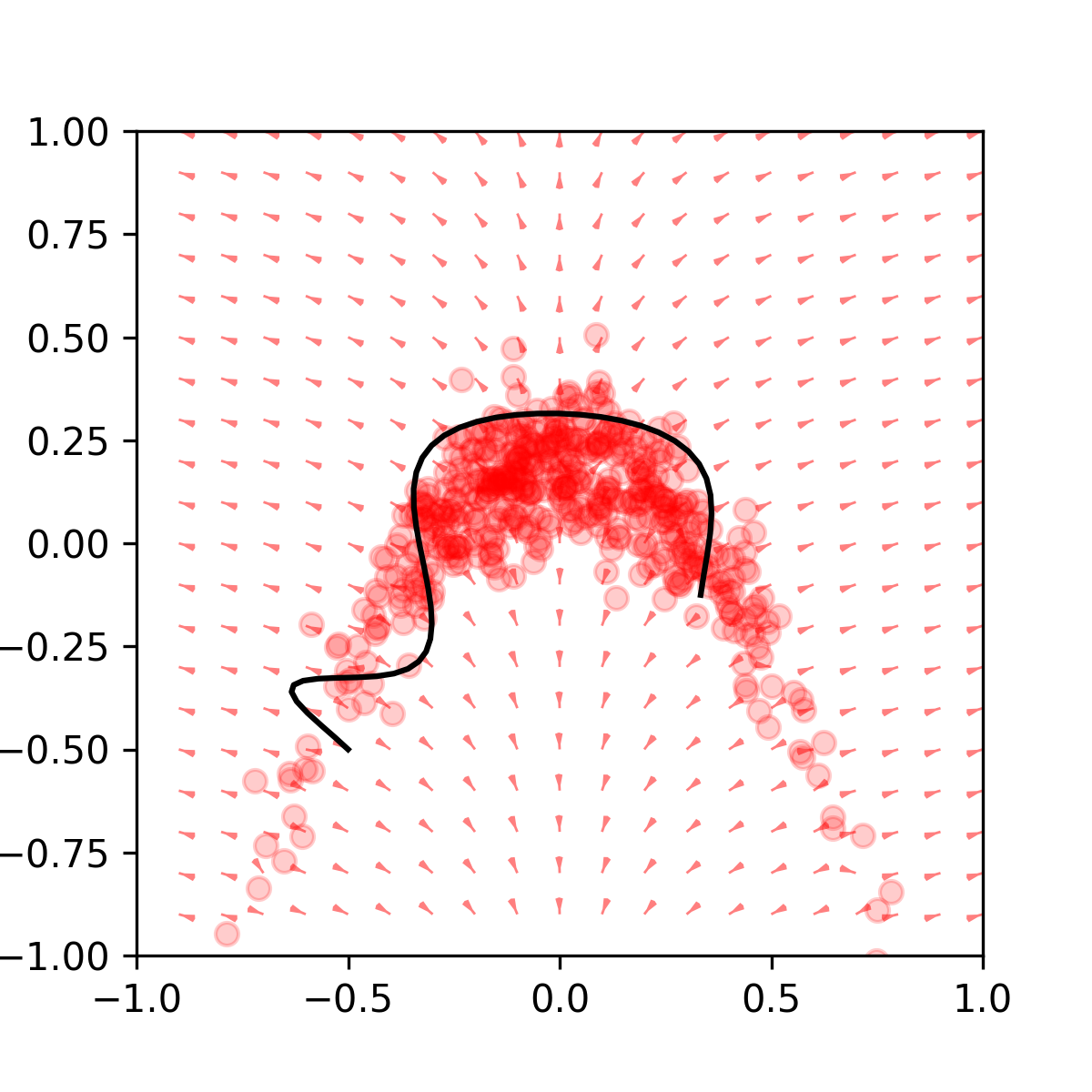}
\caption{Gradient fields, samples, and leapfrog trajectories using standard HMC (blue) and NNgHMC (red) are indistinguishable.}
\label{fig:banana}
\end{figure}

Next, we illustrate the proposed method on a \textbf{multivariate Gaussian distribution with ill-conditioned covariance.} The distribution is given by $q\sim\mathcal{N}_{30}(0,\Sigma)$ where $\Sigma$ is a diagonal matrix with smallest value $0.1$, largest value $1000$ and other values uniformly drawn between 1 and 100. As the distribution is on very disparate scales in different dimensions, HMC needs accurate gradient information to move accordingly. For HMC, the leapfrog step size $s$ is set to be 0.5 and the number of steps $l$ is set to be 100 so that acceptance probability is around 0.7. If the step size is too big, HMC would miss the high density region in the narrowest dimension. Without a sufficiently long trajectory, HMC would fail to explore the elongated tails in the widest dimension.

\begin{figure}[h!]
\centering
\includegraphics[width=0.8\linewidth, height=0.3\linewidth]{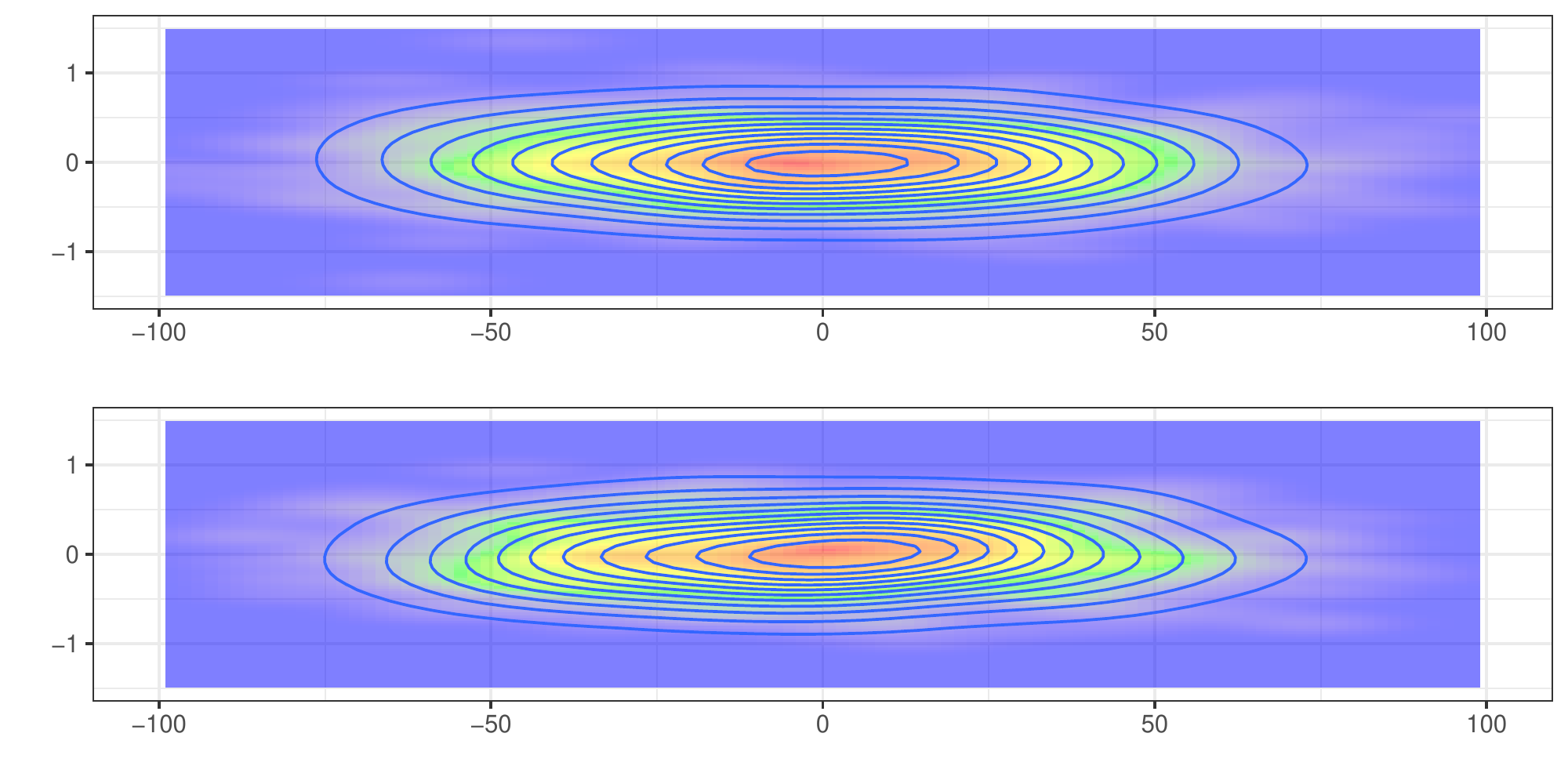}
\caption{NNgHMC posterior (bottom) captures the highly elongated shape of the Gaussian distribution in the two most extreme dimensions $(\sigma^2_1=0.1, \sigma^2_{30}=1000)$ as well as the HMC posterior (top). Note that the x- and y-axes are on very different scales.}
\label{fig:gaussian}
\end{figure}

We collect sample gradients until 50 iterations after convergence to train the neural network. The neural network has $h=100$ units in the hidden layer and is trained for $t=10$ epochs. With the same tuning parameters as standard HMC, NNgHMC has acceptance probability around 0.5. Despite slightly lower acceptance probability, as shown in Figure \ref{fig:gaussian}, NNgHMC converges to the true posterior just as standard HMC. With more training data, the neural network will learn the gradient field more accurately and NNgHMC will have similar acceptance probability as standard HMC.
 
%\begin{figure}[h!]
%\centering
%\includegraphics[width=0.8\linewidth, height=0.2\linewidth]{trajectory}
%\caption{With the same initial position and momentum, the leapfrog trajectory in the same dimensions as in Figure 1 using approximated gradient (blue) faithfully resembles the one using true gradient (red) despite heavy oscillation on the energy level contour. The periodic nature of the Hamiltonian flow reflects the fact that the Hamiltonian is that of the harmonic oscillator, i.e. both potential and kinetic energies are quadratic.   The vectors and trajectories are slightly jittered for plotting.}
%\label{fig:trajectory}
%\end{figure}

\subsection{200-dimensional Bayesian logistic regression} 

Next we demonstrate the scalability of proposed method on logistic regression with simulated data. The $X$ matrix has $50,000$ rows drawn from a 200 dimensional multivariate Gaussian distribution with mean zero and covariance $I_{200}$. The regression coefficients $\beta$ are drawn independently from $Uniform(-1,1)$. Given $X$ and $\beta$, the response vector is created with $Y_i\sim Bernoulli(logistic(X_i\beta))$. With $l=20$ leapfrog steps and step size $s=0.01$, HMC makes 1000 draws in 300 seconds with acceptance probability around $0.8$. 4000 training points and gradients, which come from 200 draws after convergence, are used for neural network training.

With the same tuning parameters, NNgHMC can make 1000 draws in just 40 seconds with acceptance probability around $0.6$. HMC yields 1.5 effective samples per second while NNgHMC yields 6.75 effective draws per second on average. The improvement on effective sample size and CPU time ratio is considerable and will only increase as the size of the data set increases.

\begin{figure}[h!]
\centering
\includegraphics[width=0.8\linewidth]{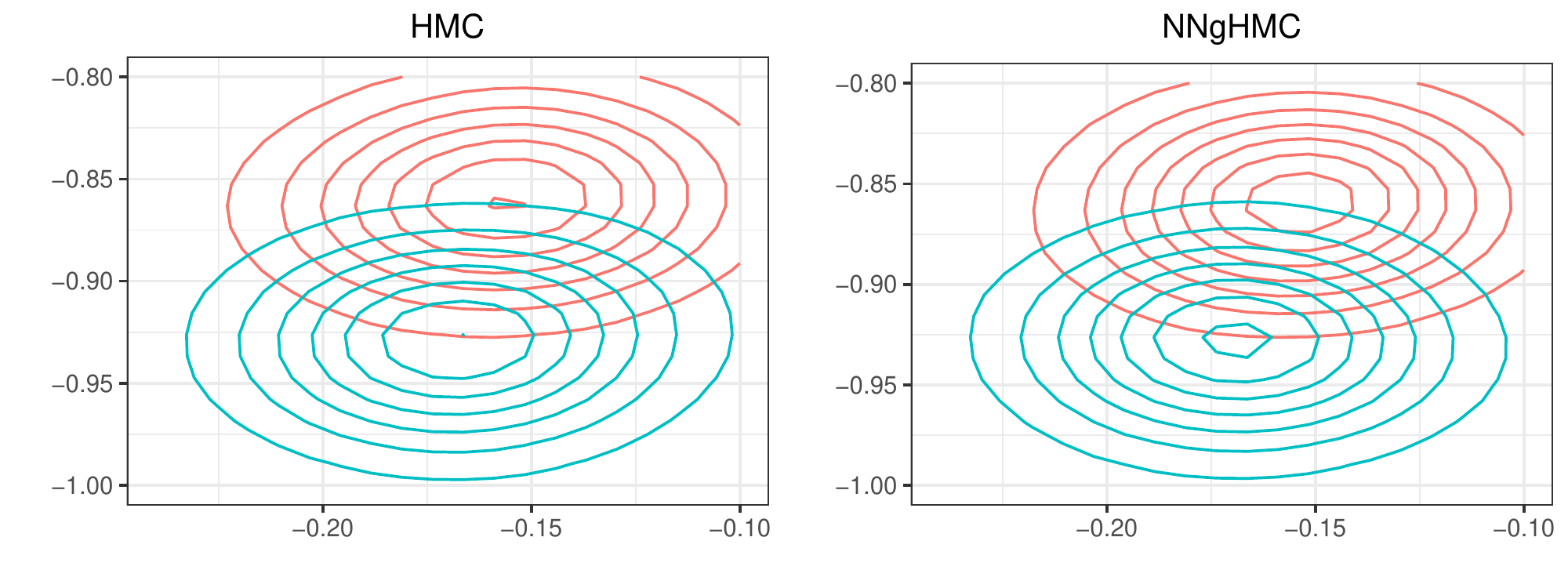}
\caption{We first use HMC to collect training samples from the posterior of the 200-dimensional logistic regression model under a diffused prior with variance 10 for NNgHMC. The HMC (left) and NNgHMC (right) posteriors are colored in green. Then we use the same trained network for NNgHMC under a concentrated prior with variance 0.1. The new HMC and NNgHMC posteriors are colored in red. \textit{Although most of the training data come from the green region, the neural network can extrapolate well to sample around the red region.}}
\label{fig:priors}
\end{figure}

The choice of prior plays an important role in Bayesian inference, and it is common to fit models with different priors for sensitivity analysis. The gradient of energy function $\nabla U$ is equal to the sum of the gradient of negative log-likelihood $-\nabla \log \pi(x|q)$ and the gradient of log prior $\nabla \log\pi(q)$. As the proposed method provides an accurate approximation of $\nabla U$ under prior $\pi(q)$, adding $\nabla \log \pi'(q)-\nabla \log \pi(q)$ will yield an approximation of $\nabla U$ under a new prior $\pi'(q)$. In this case, NNgHMC can sample from the new posterior much faster than HMC without additional training. Figure \ref{fig:priors} compares the NNgHMC and HMC posteriors.

\begin{remark}
While there are no fixed rules on the size of hidden layers, non-generative models typically have larger hidden layers than output layers. With input and output dimensions both being $200$, a large hidden layer of size $400$ would lead to $160,000$ total units, which is computationally expensive. Meanwhile, a network with a hidden layer of size $200$ has half as many total units but is not nearly as expressive. Here we use eight disjoint hidden layers of size 50 to approximate 25 dimensional blocks of the gradient to cut down the number of total units to 90,000. Figure \ref{fig:loss} compares the training losses of these three networks.
\end{remark}

\begin{figure}[h!]
\centering
\includegraphics[width=0.4\linewidth, height=0.3\linewidth]{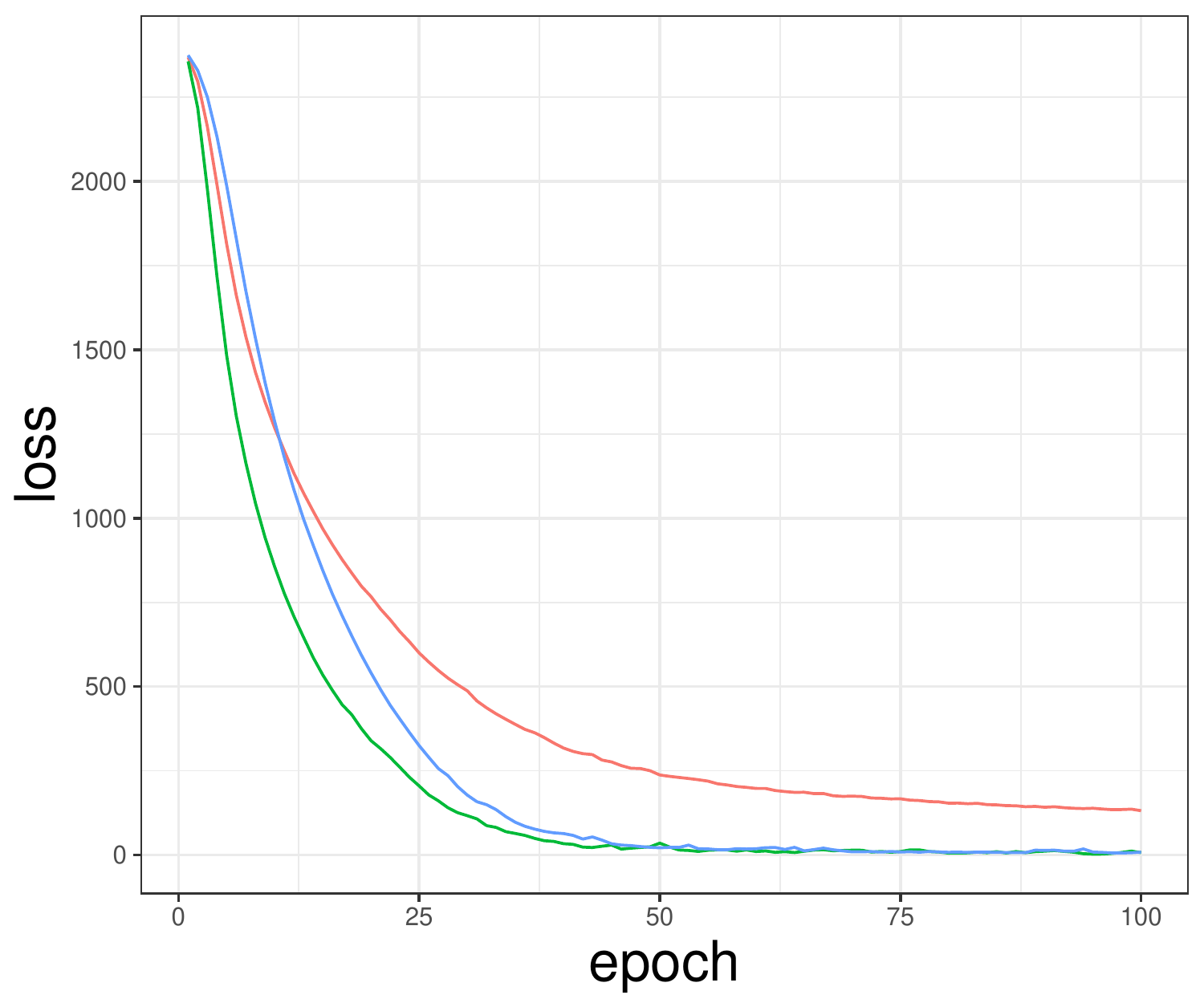}
\caption{The gradient of the 200-dimensional logistic regression model is approximated by neural networks of different designs. In terms of performance measured by training $L_2$ loss on the true gradient, the block network (blue) matches the single large network (green) and outperforms the single small network (red) using comparable number of total units.}
\label{fig:loss}
\end{figure}

\subsection{Low-dimensional models with expensive gradients}
In this section, we evaluate our method using two models that involve costly gradient evaluations in spite of their  typically low dimensions. First, we focus on the \textbf{generalized autoregressive conditional heteroskedasticity (GARCH)}, which is a common econometric model that models the variance at time $t$ as a function of previous observations and variances. The general GARCH$(m,r)$ model is given by
\begin{align}
    y_t&\sim N(0,\sigma_t^2)\\
    \sigma_t^2&=\alpha_0+\sum_{j=1}^{m}\alpha_jy_{t-j}^2+\sum_{j=1}^{r}\beta_j\sigma_{t-j}^2.
\end{align}
Conditioning on the first $\max(m,r)$ observations, the likelihood is the product of $N(0,\sigma_t^2)$ densities. The likelihood and gradient calculation for GARCH models can be slow as it has to be done iteratively and scales with the number of observations. As shown in figure \ref{fig:timeseries}, 1000 observations are generated with a $GARCH(2,1)$ model and used as data for comparing standard HMC and NNgHMC. Truncated uninformative Gaussian priors are used because of GARCH stationarity constraints. 10000 draws are taken with standard HMC and gradient values collected between 1000 to 2000 iterations are used for training. Training a neural network with hidden layer size 50 takes 5s. With tuning parameters fixed at step size $s=0.002$ and $l=15$ leapfrog steps, standard HMC and NN gradient HMC both have close to 0.7 acceptance probability, but the latter is more computationally efficient (Table \ref{tab:garch}). %Figure \ref{fig:garch} compares the marginal posteriors using standard HMC and NNgHMC draws.

\begin{table}[h!] \centering 
  \caption{Comparing standard HMC and NNgHMC using a GARCH model.} 
  \label{tab:garch} 
\begin{tabular}{cccccc} 
\\[-1.8ex]\hline 
\hline \\[-1.8ex] 
Method & AP & ESS & CPU time & Median ESS/s & Speed-up \\ 
\hline \\[-1.8ex] 
Standard  & 0.72 & (99, 261, 424) & 436s & 0.60 & 1 \\
NNg & 0.70 & (116, 176, 303) & 59s & 2.98 & \textbf{4.98} \\
\hline \\ [-1.8ex] 
\multicolumn{6}{l}{AP: acceptance probability}\\
\multicolumn{6}{l}{ESS: effective sample size (min, median, max) after removing 10\% burn-in}\\
\hline \hline
\end{tabular}
\end{table} 

\begin{figure}[h!]
\centering
\includegraphics[height=50px]{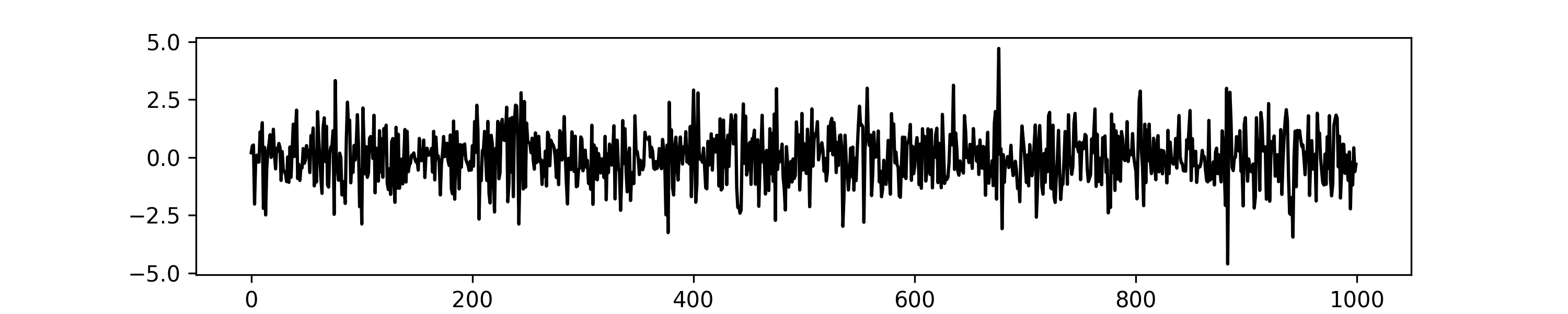}
\caption{Time series data generated with a GARCH(2, 1) model.}
\label{fig:timeseries}
\end{figure}

% \begin{figure}[h!]
% \centering
% \includegraphics[width=0.49\linewidth]{garch21_hmc_posterior.png}
% \includegraphics[width=0.49\linewidth]{garch21_nnghmc_posterior.png}
% \caption{GARCH(2, 1) model marginal posteriors using standard HMC (blue) and NNgHMC draws (red).}
% \label{fig:garch}
% \end{figure}

\textbf{Gaussian process} is computationally expensive because the covariance matrix is $n\times n$ and inverting it requires $\mathcal{O}(n^3)$ computation. Here we consider a Gaussian process regression model with the Mat\'ern kernel:

\begin{align}
    Y&\sim\mathcal{N}(0,\mathcal{K}(X,X))\\
    \mathcal{K}(d)&=\frac{2^{1-\nu}}{\Gamma(\nu)}(\sqrt{2\nu}\frac{d}{l})^\nu K_\nu(\sqrt{2\nu}\frac{d}{l})
\end{align}
where $d$ is the Euclidean distance between two observations $x$ and $x'$, length scale $l$ and smoothness $\nu$ are the hyper-parameters. $\Gamma$ denotes the gamma function and $K_\nu$ is the modified Bessel function of the second kind. Here $\nu$ is fixed at 1.5 to limit Gaussian process draws to be once differentiable functions. It is common to add white noise $\sigma^2 I$ to the covariance matrix for numerical stability. Therefore, the second free hyper-parameters besides $l$ is $\sigma^2$. Diffused Lognormal priors are used for the hyperparameters. The $500\times4$ data matrix $X$ is drawn from Multivariate Gaussian with mean zero and identity covariance. $Y$ is obtained by mapping $X$ with a polynomial pattern and adding noise.

10000 draws are sampled using standard HMC with leapfrog steps $l=20$ and step size $s=0.05$; the acceptance probability is 0.83 but it is very time consuming. Gradient collected during the first 1000 draws is then used to train a neural network with hidden layer size $h=100$ for $t=100$ epochs. Using the same tuning parameters, NNgHMC can sample 10000 draws in much shorter time with the same acceptance probability (Table \ref{tab:gp}). Figure \ref{fig:regression} compares the standard HMC and NNgHMC posteriors; Figure \ref{fig:band} compares Gaussian process model posterior draws along one particular direction.

\begin{table}[h!] \centering 
  \caption{Experiment results using Gaussian process regression model} 
  \label{tab:gp} 
\begin{tabular}{cccccc} 
\\[-1.8ex]\hline 
\hline \\[-1.8ex] 
Method & AP & ESS & CPU time & Median ESS/s & Speed-up \\ 
\hline \\[-1.8ex] 
Standard & 0.83 & (5135, 5754, 7635) & 1834s & 3.14 & 1 \\
NNg & 0.84 & (4606, 6172, 7741) & 50s & 123.4 & \textbf{39.3} \\
\hline \\ [-1.8ex] 
\multicolumn{6}{l}{AP: acceptance probability}\\
\multicolumn{6}{l}{ESS: effective sample size (min, median, max) after removing 10\% burn-in}\\
\hline \hline
\end{tabular}
\end{table} 

\begin{figure}[h!]
\centering
\includegraphics[width=0.4\linewidth]{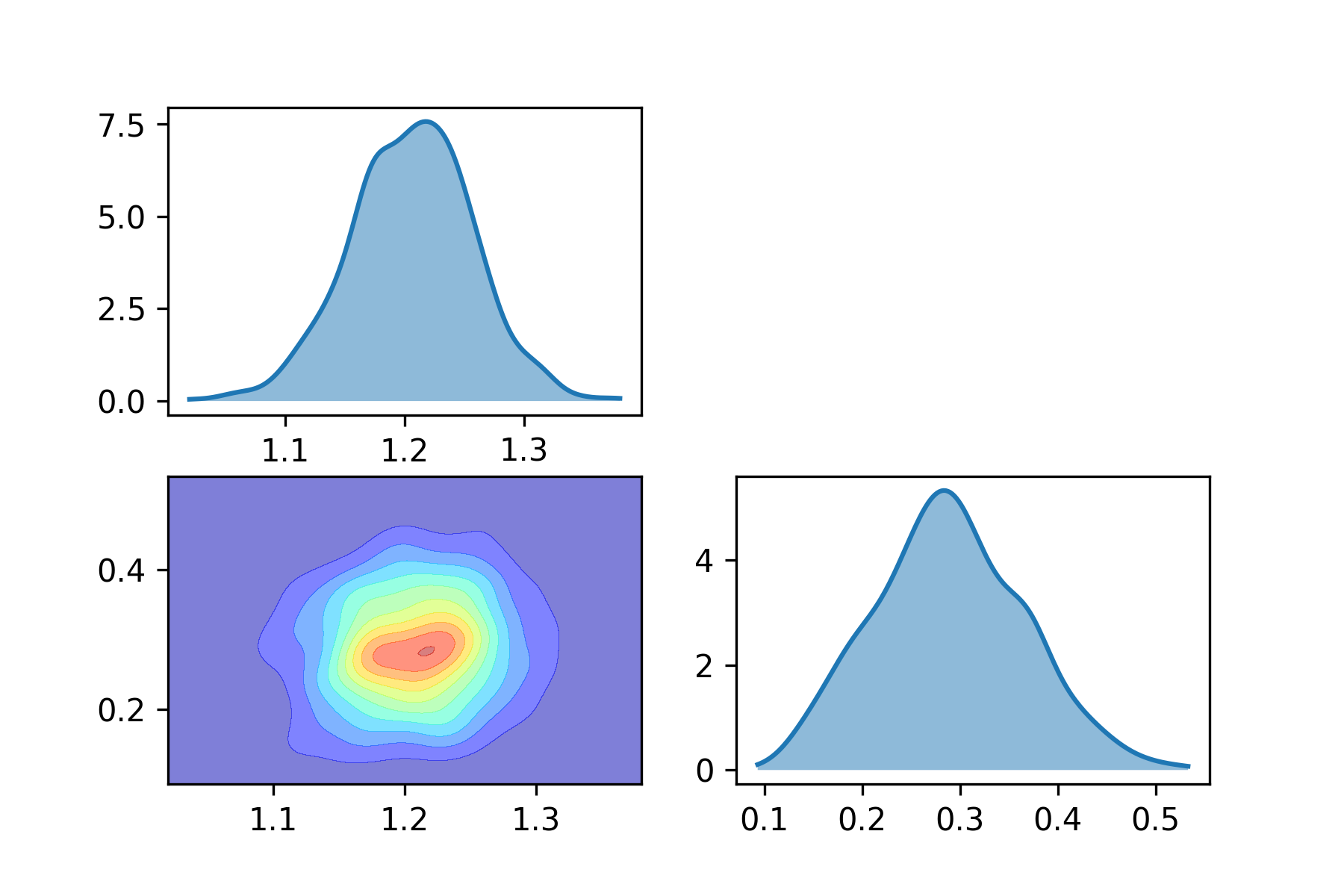}
\includegraphics[width=0.4\linewidth]{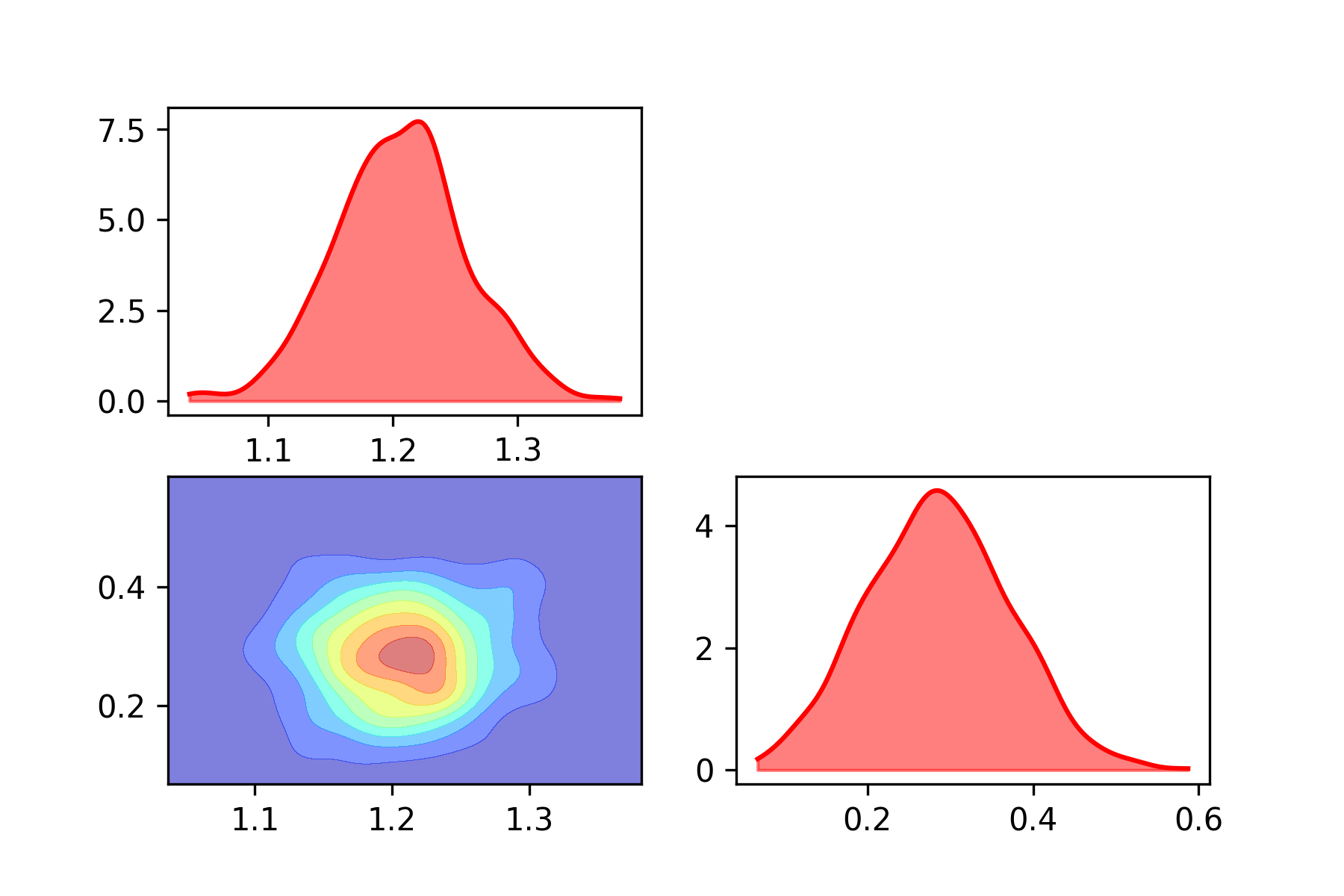}
\caption{GP regression model posteriors of hyper-parameters using standard HMC (blue) and NNgHMC draws (red).}
\label{fig:regression}
\end{figure}

\begin{figure}[h!]
\centering
\includegraphics[width=0.4\linewidth]{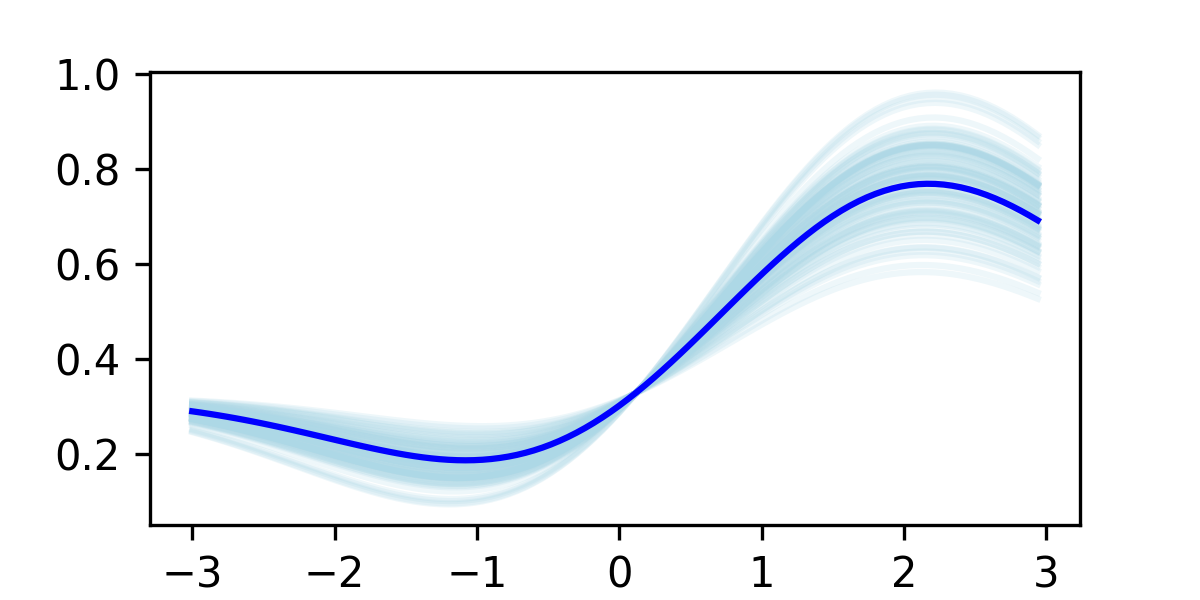}
\includegraphics[width=0.4\linewidth]{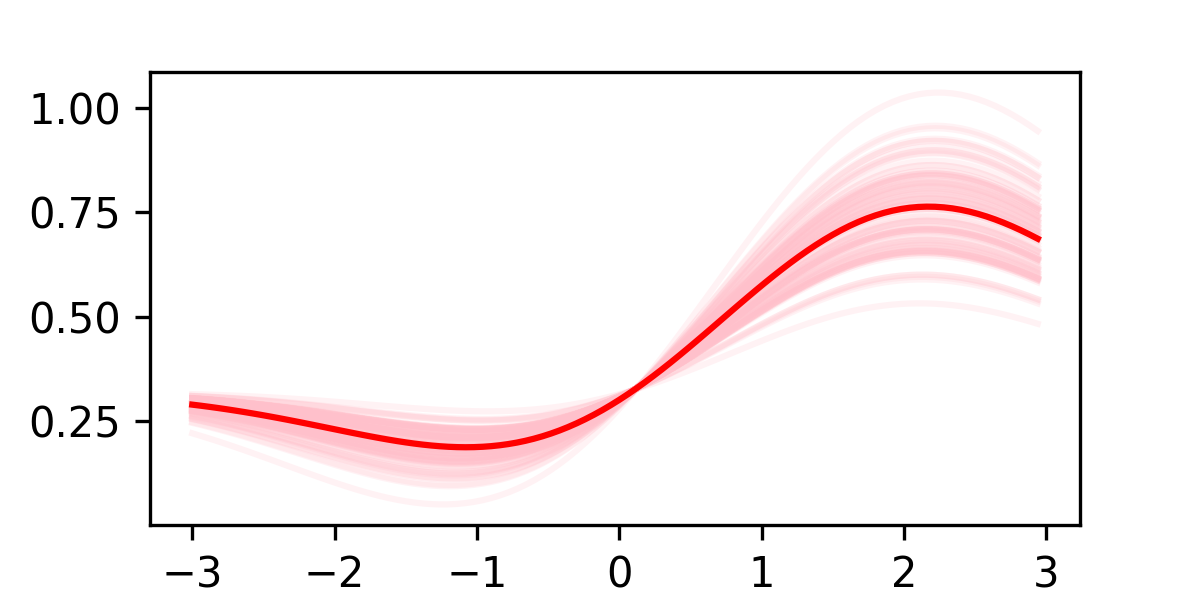}
\caption{GP regression model predictions with standard HMC (blue) and NNgHMC posteriors (red).}
\label{fig:band}
\end{figure}

\subsection{Comparison with stochastic gradient HMC}

Na\"ive stochastic gradient HMC using mini-batches of data is problematic as the noisy gradient can push the sampler away from the target region. Recent more advanced stochastic gradient method uses a friction term and is shown to sample from the true posterior asymptotically. The formulation of SGHMC is given by:

\begin{align}
d\theta &= M^{-1}r dt\\
dr&=-\nabla U(\theta)dt-BM^{-1}rdt+N(0,2Bdt)
\end{align}

where $N(0, 2Bdt)$ is the noise added to the gradient by subsampling. In practice, the friction term $BM^{-1}rdt$ is set arbitrarily.

To further improve speed, SGHMC does not perform Metropolis-Hastings correction and uses very small step sizes. The SGHMC posterior is dependent on the choice of step size; however, \textit{a priori} one would not know the optimal step size. Here we want to show that while SGHMC provides fast approximation of the true posterior when data are abundant, the SGHMC posterior may not be suitable for inference. 

In our experiment, the Cover Type data from UCI machine learning repository is used. We run standard HMC for 4000 iterations with $l=50$ leapfrog steps and step size $s=0.002$. We also run SGHMC for 4000 iterations with default parameters and varying step sizes from $s=5e-6$ to $s=5e-8$.

Figure \ref{fig:sghmc} illustrates the main issue with SGHMC. For these two marginal distributions, the SGHMC posteriors have roughly the same location but completely different shapes. On the other hand, NNgHMC posteriors agree with the standard HMC posteriors almost exactly.

\begin{figure}[h!]
\centering
\includegraphics[width=0.4\linewidth]{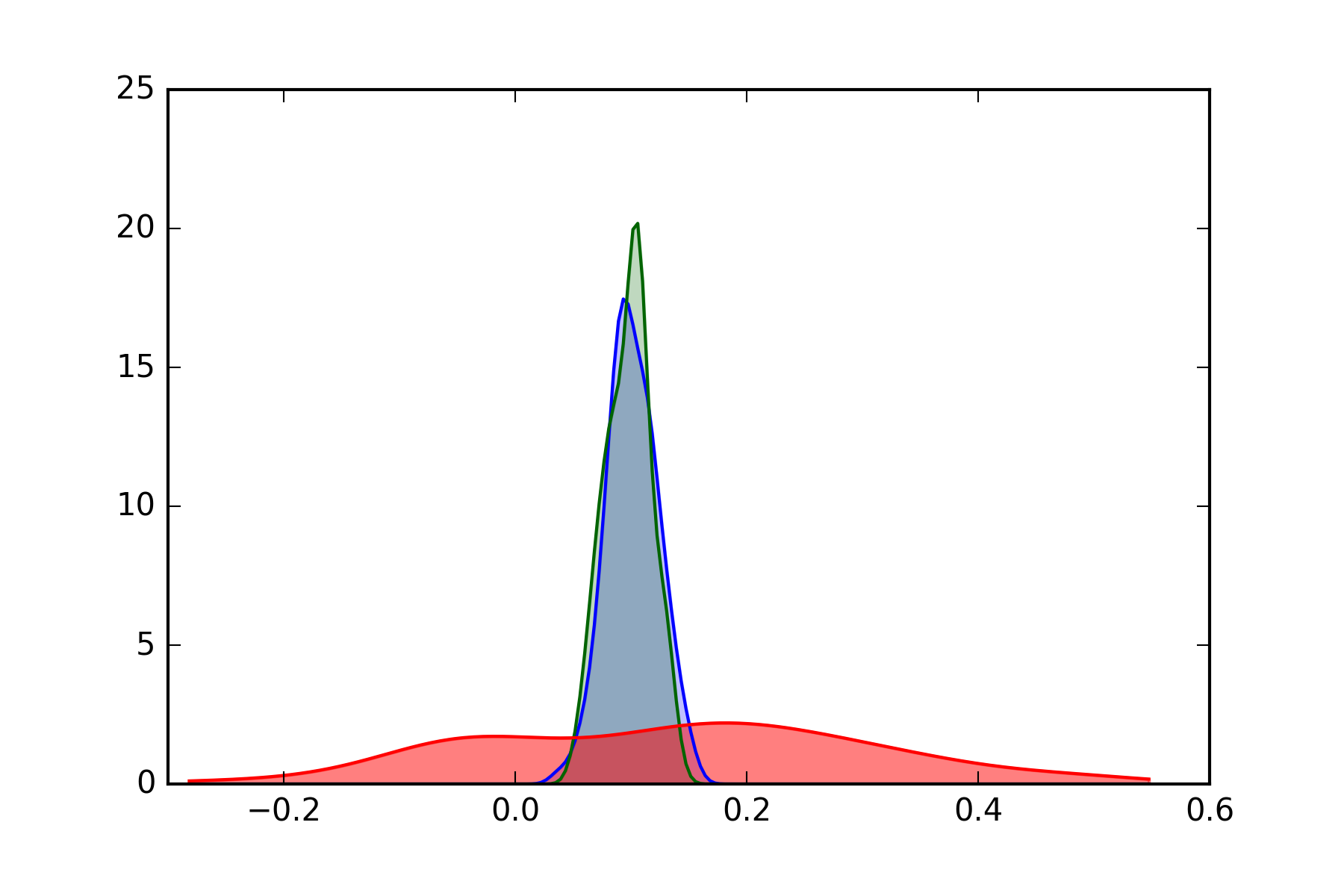}
\includegraphics[width=0.4\linewidth]{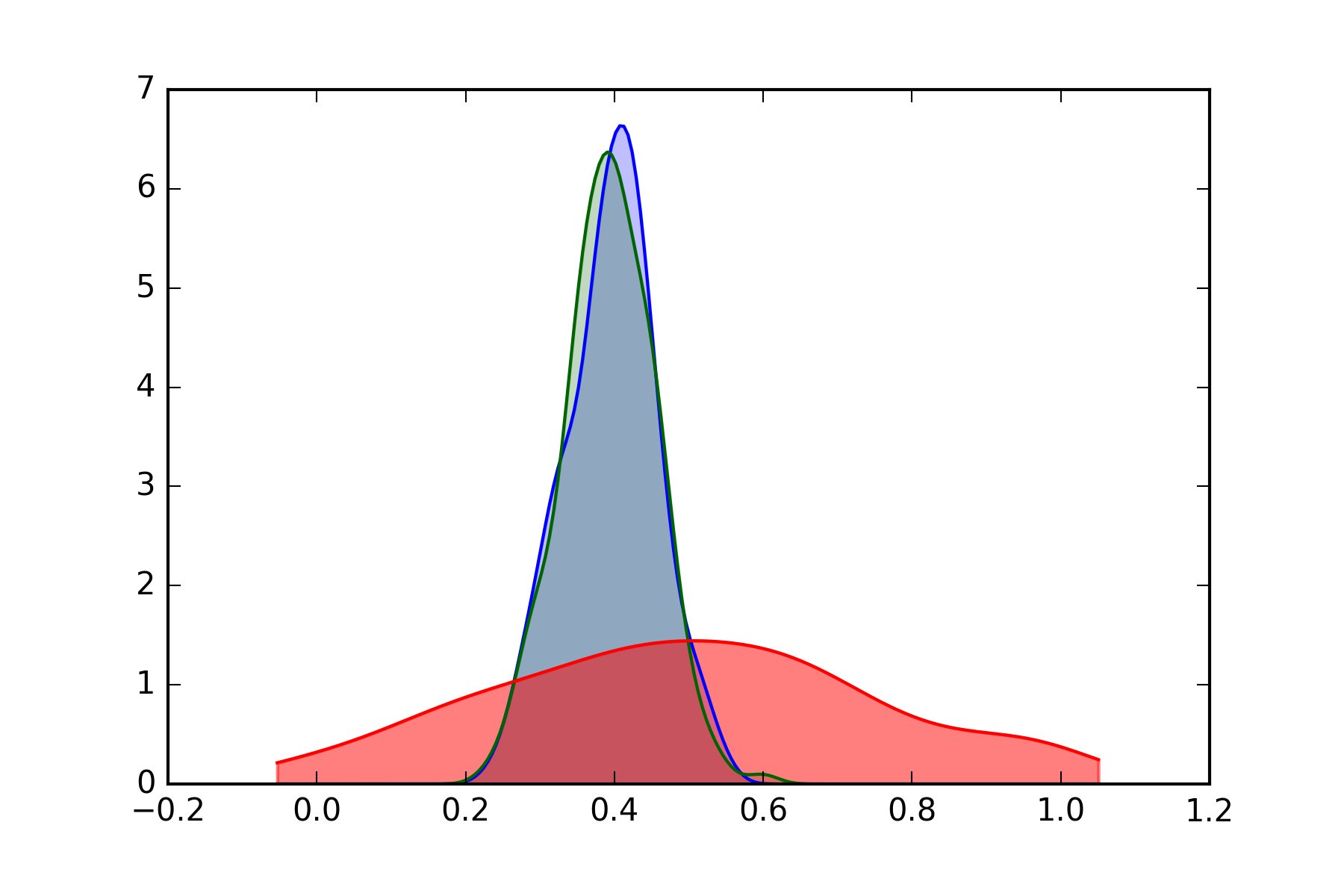}
\caption{Histograms of marginal posteriors of logistic regression model coefficients with Laplace prior on Cover Type data. Blue: standard HMC; Red: stochastic gradient HMC; Green: neural network gradient HMC}
\label{fig:sghmc}
\end{figure}

Another comparison with SGHMC is performed with Metropolis-Hastings correction. Here the sub-sampled data size is 5000 and the tuning parameters are $l=10$ leapfrog steps and step size $s=0.001$ so that the simulated trajectory is shorter for less gradient noise to compound. While SGHMC is faster still, the quality of samples is inferior compared to proposed method as indicated by lower ESS in Table \ref{tab:cover} and less mixed trace plot in Figure \ref{fig:cover}. Overall, NNgHMC still outperforms SGHMC in terms of median EES per second. 

\begin{table}[h!] \centering 
  \caption{Experiment results on Cover Type data} 
  \label{tab:cover} 
\begin{tabular}{cccccc} 
\\[-1.8ex]\hline 
\hline \\[-1.8ex] 
Method & AP & ESS & CPU time & Median ESS/s & Speed-up \\ 
\hline \\[-1.8ex] 
Standard & 0.80 & (73, 143, 10000) & 3147s & 0.05 & 1 \\
NNg & 0.67 & (57, 186, 7174) & 710s & 0.26 & \textbf{5.77} \\
SG & 0.33 & (49, 59, 246) & 357s & 0.17 & 3.64 \\
\hline \\ [-1.8ex] 
\multicolumn{6}{l}{AP: acceptance probability}\\
\multicolumn{6}{l}{ESS: effective sample size (min, median, max) after removing 10\% burn-in}\\
\hline \hline
\end{tabular}
\end{table} 

\begin{figure}[H]
\centering
\includegraphics[width=0.8\linewidth]{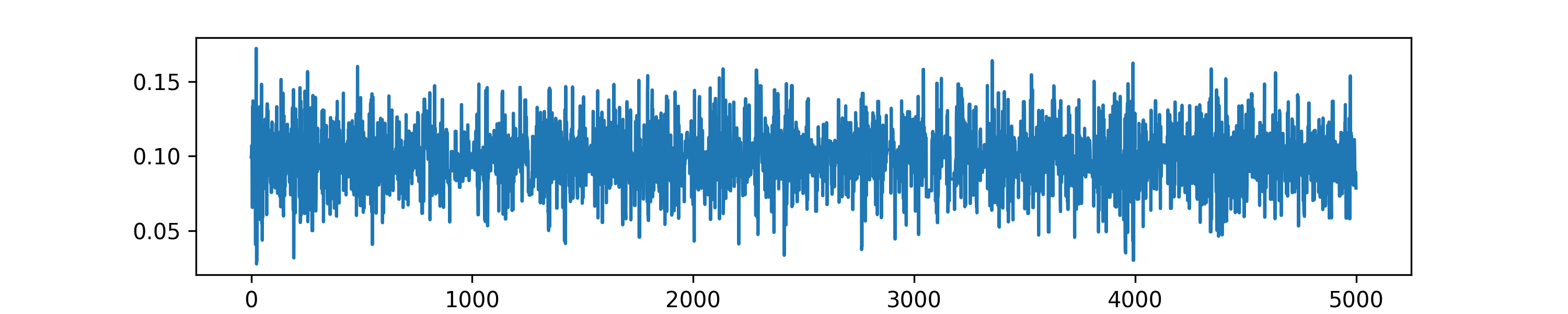}
\includegraphics[width=0.8\linewidth]{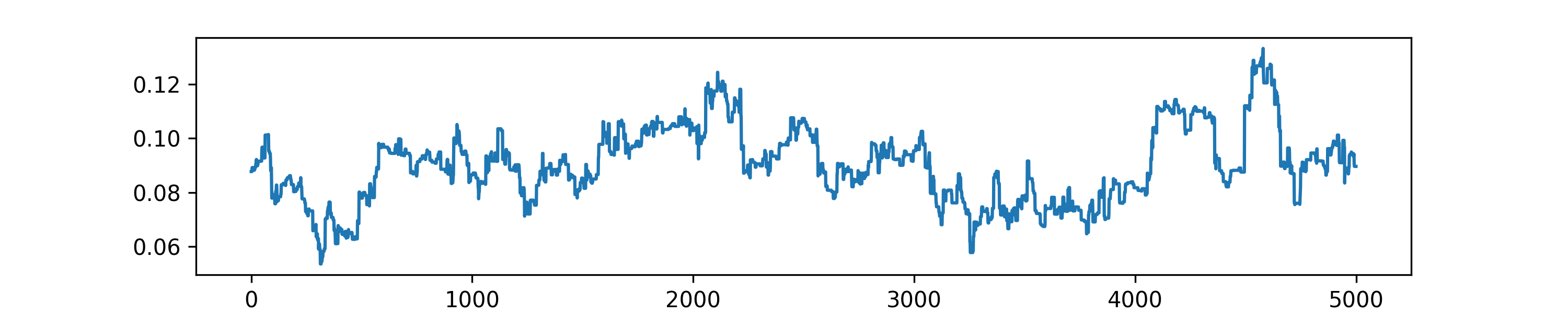}
\caption{Trace plots of NNgHMC (top) and stochastic gradient HMC (below) show that the NNgHMC chain is more efficient as the approximated gradient is more accurate than sub-sampled gradient.}
\label{fig:cover}
\end{figure}

\subsection{Comparison with Gaussian process surrogate}

We now compare our method to the Gaussian process surrogate approach with the squared exponential kernel parametrized by $\mathcal{K}(x,x')=\exp{-\frac{(x-x')^2}{2l^2}}$. We also add white noise $\sigma^2 I$ to the covariance matrix. The squared exponential kernel, the default choice, is infinitely differentiable and gives rise to another Gaussian process as the derivative. Given observations $X$ and $Y$, we can explicitly write down the mean of the derivative at $x^*$.

\begin{align}
E\frac{\partial f}{\partial x^*}=\frac{\partial}{\partial x^*} E f=\frac{\partial}{\partial x^*}\mathcal{K}(x^*,X)\mathcal{K}(X,X)^{-1}Y
\end{align}

Here we estimate both the length scale $l$ in the squared exponential kernel and the white noise parameter $\sigma$ jointly with maximum likelihood. The estimation requires inverting the observed covariance matrix, which is $\mathcal{O}(n^3)$ where $n$ is the number of observations.

We compare with GP surrogate method on multivariate Gaussian distributions with covariance $I_d$ where $d$ varies from 10 to 40. For each Gaussian, we generate $n$ training data points to train the GP surrogate and neural network. The neural network has 100 hidden units and is trained for 10 epochs. After training, both methods are used to draw 1000 samples. 

Table 4 compares acceptance probability for the two methods. We can see that the neural network predicted gradient provides better approximation overall than the gradient of GP as indicated by higher acceptance probability. This advantage is more pronounced as dimensionality increases. 

\begin{table}[H] \centering 
  \caption{Acceptance probability when sampling from multivariate Gaussian} 
  \label{} 
\begin{tabular}{ccccc} 
\\[-1.8ex]\hline 
\hline \\[-1.8ex] 
Method & Dimension / Training & 500 & 1000 & 2000 \\ 
\hline \\[-1.8ex] 
Gaussian process & 10 & 0.65 & 0.61 & 0.57 \\
& 20 & 0.64 & 0.65 & 0.62 \\
& 40 & 0.31 & 0.32 & 0.32 \\
\hline \\[-1.8ex] 
Neural network gradient & 10 & 0.95 & 0.96 & 0.97 \\
& 20 & 0.82 & 0.87 & 0.91\\
& 40 & 0.61 & 0.75 & 0.87 \\
\hline \hline
\end{tabular}
\end{table}

\subsection{Speed evaluation on real data}
Similar to other surrogate methods, NNgHMC has three stages: training data collection, training, and sampling. We have demonstrated that using a neural network can provide accurate approximation of the gradient, however, the effectiveness of our method still needs to be evaluated by time. If the neural network requires too much training data, then it would not reduce computation time. Here we first run standard HMC to draw a desired number of samples (10000) and record time as benchmark. Then we collect different amounts of training data points (10\%, 15\% and 20\% of total number) and use NNgHMC to draw remaining samples. The time to collect training data and train the neural network is included for NNgHMC. As shown in Table 5, 10\% of training data is sufficient for the neural network to learn the gradient and gives the most speed-up. While adding more training data increases the quality of gradient approximation, the computation cost outweighs the benefit of higher acceptance probability. 

\begin{table}[H] \centering 
  \caption{Experiment results on data sets from UCI machine learning repository} 
  \label{} 
\begin{tabular}{cccccc} 
\\[-1.8ex]\hline 
\hline \\[-1.8ex] 
Method & AP & ESS & CPU time & Median ESS/s & Speed-up \\ 
\hline \\[-1.8ex] 
\multicolumn{6}{c}{Online News Popularity ($39797\times44$)}\\
\hline \\[-1.8ex] 
Standard & 0.77 & (777, 2021, 5929) & 3607s & 0.66 & 1 \\
NNg (10\%) & 0.61 & (605, 1416, 4865) & 502s & 2.82 & \textbf{4.27} \\
NNg (15\%) & 0.64 & (620, 1382, 5500) & 678s & 2.04 & 3.09 \\
NNg (20\%) & 0.68 & (700, 1731, 5397) & 854s & 2.03 & 3.08 \\
\hline \\ [-1.8ex] 
\multicolumn{6}{c}{Census Income ($48842\times123$)}\\
\hline \\[-1.8ex] 
Standard & 0.84 & (6306, 9390, 10000) & 1796s & 5.23 & 1 \\
NNg (10\%) & 0.60 & (4023, 6024, 7156) & 564s & 10.68 & 2.04 \\
NNg (15\%) & 0.68 & (4617, 7511, 9201) & 656s & 11.45 & \textbf{2.19} \\
NNg (20\%) & 0.76 & (5036, 7558, 8696) & 740s & 10.21 & 1.95 \\
\hline \\ [-1.8ex] 
\multicolumn{6}{c}{Dota2 Games Results ($102944\times116$)}\\
\hline \\[-1.8ex] 
Standard & 0.75 & (1677, 5519, 8621) & 20760s & 0.27 & 1 \\
NNg (10\%) & 0.59 & (1446, 4197, 6442) & 2903s & 1.45 & \textbf{5.44} \\
NNg (15\%) & 0.70 & (1901, 4865, 7600) & 3911s & 1.24 & 4.59 \\
NNg (20\%) & 0.74 & (2432, 5744, 8860) & 4992s & 1.15 & 4.26 \\
\hline \\ [-1.8ex] 
\multicolumn{6}{l}{AP: acceptance probability}\\
\multicolumn{6}{l}{ESS: effective sample size (min, median, max) after removing 10\% burn-in}\\
\hline \hline
\end{tabular}
\end{table} 

\section{Discussion}

Whereas HMC is helpful for computing large Bayesian models, its repeated gradient evaluations become overly costly for big data analysis.  We have presented a method that circumvents the costly gradient evaluations, not by subsampling data batches but by learning an approximate gradient that is functionally free of the data.  We find that multi-output, feedforward neural networks are ripe for this application: NNgHMC is able to handle models of comparatively large dimensionality.%, with as much as a 40-fold boost in efficiency. In short, Bayesians can have big models and compute them too.

The NNgHMC algorithm is an important paradigm shift away from the class of surrogate function approximate HMC algorithms, but this shift leaves many open questions. Much work is needed to extend NNgHMC to an on-line, adaptive methodology: what measures of approximation error will be useful criteria for ending the training regime of the algorithm, and are there benefits to iterating between training and sampling regimes? Are there any valid second-order extensions to the NNgHMC algorithm \`a la Riemannian HMC? Finally---and most interestingly---can the representational power of deep neural networks be leveraged for more accurate approximations to the Hamiltonian flow? 

\newpage

\appendix

\bibliographystyle{plain}
\bibliography{paper}

\end{document}